\documentclass[a4paper, 10pt, twosides]{amsart}
\usepackage{defs}

\title[MoQ of finite-dimensional linear systems]{Measure of quality of finite-dimensional linear systems: a frame-theoretic view}
\author[Mishal Assif PK]{Mishal Assif PK}
\address{Department of Mechanical Engineering\\ IIT Bombay, Powai\\ Mumbai 400076, India.}
\author[M. Rayyan Sheriff]{Mohammed Rayyan Sheriff}
\address{Systems \& Control Engineering\\ IIT Bombay, Powai\\ Mumbai 400076, India.}
\author[D. Chatterjee]{Debasish Chatterjee}
\address{Systems \& Control Engineering\\ IIT Bombay, Powai\\ Mumbai 400076, India.}
\thanks{Emails: \texttt{mishal\_assif@iitb.ac.in, rayyan@iitb.ac.in, dchatter@iitb.ac.in.}}
\thanks{The authors thanks K.\ S.\ Mallikarjuna Rao, IIT Bombay, India, for a pointer to \cite{ref:Lio-94}, and M.\ Nagahara, University of Kitakyushu, Japan, and T.\ Ohtsuka, Kyoto University, Japan, for their encouraging comments on this topic and the associated results.}

\begin{document}

\begin{abstract}
	 A measure of quality of a control system is a quantitative extension of the classical binary notion of controllability. In this article we study the quality of linear control systems from a frame-theoretic perspective. We demonstrate that all LTI systems naturally generate a frame on their state space, and that three standard measures of quality involving the trace, minimum eigenvalue, and the determinant of the controllability Gramian achieve their optimum values when this generated frame is tight. Motivated by this, and in view of some recent developments in frame-theoretic signal processing, we propose a natural measure of quality for continuous time LTI systems based on a measure of tightness of the frame generated by it and then discuss some properties of this frame-theoretic measure of quality.
\end{abstract}

\maketitle

\section{Introduction}

\label{section:Review-of-some-conventional-MOQs-for-LTI-systems}

Let \(n\) and \(m\) be positive integers, and consider a linear time-invariant (LTI) system
\begin{equation}
\label{eq:lti-system}
\dot{x}(t) = A x(t) + B u(t) \quad \text{for \(t \in \R{}\),}
\end{equation}
where \( A \in \mat{n}{n} \) and \( B \in  \mat{n}{m} \) are given and fixed matrices. Let \( \horizon > 0 \) denote the time horizon of the system \eqref{eq:lti-system} for the control objectives that follow, and let \( \controlprofile \) denote the set of control maps \( u : [0, \horizon] \longrightarrow \R{m} \) that are square integrable. The \emph{reachability matrix} \( \kalman \) and the \emph{reachable space} \( \mathcal{R} \) of LTI system \eqref{eq:lti-system} are defined by:
\begin{align*}
\kalman & \coloneqq \pmat{B}{AB}{A^{n - 1}B}, \\
\mathcal{R} & \coloneqq \image ( \kalman ).
\end{align*}
Recall that the control system \eqref{eq:lti-system} is \emph{controllable} in the classical sense if, given any preassigned points of \( \bar x, \hat x\in\R{\sysdim} \), there exists a control \( u \in \controlprofile \) that can transfer the states of the system \eqref{eq:lti-system} from \( x(0) = \bar x \) to \( x(\horizon) = \hat x \). It is known that the LTI system \eqref{eq:lti-system} is controllable if and only if the matrix \( \kalman \) is of rank \( n \). Of course, there are analogues of controllability for nonlinear and stochastic systems, and each of these notions provides a certificate of whether the corresponding control system, locally or globally, is controllable or not.

In our everyday lives, in addition to knowing whether a control system is controllable or not, it is also important to understand \emph{how controllable} or \emph{how good} is the control system. Indeed, a person intending to purchase a car typically test-drives several models within the budget, and in addition to its efficiency, the ease of maneuverability of the car, its ability to handle tight corners at various speeds, etc., during the test-drives become important factors in arriving at the final selection. Similarly, a surgeon performing a robot-assisted surgery would naturally prefer the instruments to be as amenable but as precise as possible in order to maximize the success of the operation. However, no assessment of the ``extent'' of controllability of a given control system is provided by the classical ideas, and in this article we propose a natural and fundamental technique to do precisely that.

Intuitively speaking, any measure of controllability of systems should involve some important and innate characteristics of the system such as the average control energy or the control effort required to perform a certain class of manoeuvres, robustness to a class of disturbances, the ability to control the system with a class of sparse controls, etc. We shall observe below that such intuitive ideas are justified, and indeed, they are natural. Our measure of controllability relies on the theory of frames, an extremely popular topic in signal processing, and unifies and sheds new light on a plethora of controllability measures that have appeared in the literature so far. Indeed, we select several measures of controllability that have been proposed across several decades in \cite{ref:MulWeb-72}, \cite{ref:Lio-94}, \cite{pasqualetti2014controllability}, \cite{ref:SumCorLyg-16}, \cite{ref:ZhaCor-17}, and unify and derive new insights into all of them under a single umbrella framework. This particular ability to collect such diverse ideas under one umbrella points strongly to the fundamental character of our framework. In order to ensure a clean and simple exposition we shall limit our discussion to the context of LTI systems as in \eqref{eq:lti-system} above, and refer to our controllability measure as a \emph{Measure Of Quality} (MOQ) of an LTI control system.

A few words about frames are in order. Frames are, roughly speaking, overcomplete bases of Hilbert spaces. The property of overcompleteness ensures that the representation of vectors in terms of frames (as opposed to bases) leads to strong robustness properties of such representations that are useful in signal processing. The study of frames was initiated by Duffin and Schaeffer \cite{duffin1952class} and expanded greatly by work of Daubechies et al.\ in \cite{daubechies1986painless}. A particular class of frames, namely, tight frames, are of great importance in signal processing. Tight frames are minimizers of a certain potential function \cite{benedetto2003finite, casazza2006physical}, they possess several desirable properties. For example, representations of vectors in terms of tight frames exhibit better resilience to noise and quantization \cite{shen2006image, zhou2016adaptive}; tight frames are also known to be good for representing signals sparsely, and are the \( \ell_2 \)-optimal dictionaries for representing vectors that are uniformly distributed over spheres \cite[Proposition 2.13, p.\ 10]{ref:SheCha-17}. The results in the sequel will demonstrate that the theory of tight frames lead to important useful consequences in control theory, in particular, in the context of quantifying controllability of LTI systems.

The rest of the article unfolds as follows. In \secref{sec: frame theory} we introduce the basic definitions associated with LTI systems with an emphasis on the underlying frame-theoretic aspects. In particular, we show how a controllable LTI system naturally gives rise to a frame on it's state space. In \secref{sec: optimal control} we review the minimum energy optimal control problem for LTI systems. We then discuss three classical MOQs motivated by the minimum energy control problem that have appeared in literature before. In \secref{sec: optimization} we present our first contribution: we show that all the three classical MOQs introduced earlier attain their optimal value precisely when the frame generated by the LTI system is tight. In \secref{sec: frame MOQ}, motivated by the results obtained in the preceding section, we propose a new measure of quality for LTI systems based on how tight the frame generated by the LTI system is. We then discuss some properties of the proposed frame-theoretic measure of quality. 

\subsection*{Notation}
We employ standard notations in this article. The set of positive integers is denoted by \(\N\), the real numbers by \(\R{}\). For any positive integer \(\nu\) and a vector \(x\) in \(\R{\nu}\), we let \( x\transp \) denote its transpose and \( \norm{x} \) its standard Euclidean norm. We work with several different inner products in the sequel, and as a rule we distinguish the inner product on a vector space \(X\) as \(\inprod{\cdot}{\cdot}_X\), with the exception of the standard inner product \(\inprod{v}{v'} = v\transp v'\) on \(\R{\nu}\) that we leave without a subscript to avoid notational clutter. The norm induced by \(\inprod{\cdot}{\cdot}_X\) on $X$ is denoted by $\|\cdot\|_{X}$, once again with the exception of the standard Euclidean norm on \(\R{\nu}\). For a matrix \(M\) with real entries, \( \image(M)\) is its column space and \( \trace(M) \) is its trace. The set of \( \nu \times \nu \) symmetric positive definite and non-negative definite matrices are denoted by \( \posdef{\nu} \) and \( \possemdef{\nu} \), respectively. The \(\nu\times\nu\) identity matrix will be denoted by \(I_\nu\). For us $\ltwo{\R{\nu}}$ stands for the Hilbert space of square summable sequences taking values in $\R{\nu}$, i.e., \(\ltwo{\R{\nu}} \Let \aset[\Big]{ \alpha = (\alpha_i)_{i\in\N} \suchthat \alpha_i\in\R{\nu}\text{ for each }i \text{ and }\sum_{j=1}^{+\infty} \norm{\alpha_j}^2 < +\infty }\), equipped with the inner product \(\inprod{\alpha}{\alpha'}_{\ltwo{\R{\nu}}} \Let \sum_{i=1}^{+\infty} \inprod{\alpha_i}{\alpha'_i}\).

\section{Frame theory and Linear systems}
\label{sec: frame theory}

We recall the basic definitions related to frames in finite-dimensional Hilbert spaces:
\begin{definition}[{\cite[\S1.1]{ref:Chr-16}}]
	\label{d: frame}
	For a Hilbert space  \( H_n \) of dimension \(n\) and with an inner product \( \langle \cdot , \cdot \rangle_{H_n} \), a finite or countable collection of vectors \( (\ltiframeelement{i})_{i \in \indexset} \subset H_n \) is said to be a \emph{frame} of \( H_n \) if there exist constants \( 0 < c \leq C \) such that 
	\[
		c \norm{v}_{H_n}^2 \leq \sum_{i \in \indexset} \abs[\big]{ \inprod{v_i}{v}_{H_n} }^2 \leq C \norm{v}_{H_n}^2 \quad \text{for all \( v \in H_n \)}.
	\]
	A frame is said to be \emph{tight} if \( c = C \).
\end{definition}

We emphasize that the index set $\indexset$ in Definition \ref{d: frame} can be finite or countably infinite. Elementary arguments show that a countably infinite collection of vectors \( (\ltiframeelement{i})_{i \in \indexset} \) of vectors is a frame of \( H_n \) if and only if it is square summable and spans \( H_n \). However, for a finite collection of vectors to constitute a frame, it suffices that they span \( H_n \). 

\begin{figure}[H]
			\centering
			\subfloat[]{
			\begin{tikzpicture}
			\coordinate (A) at (0,0);
			\coordinate (B) at (0,1);
			\coordinate (C) at (-0.866, -0.5);
            \coordinate (D) at (0.866, -0.5);
            \coordinate (E) at (1,0);
            \coordinate (F) at (-0.707, 0.707);
			
			\draw [fill=blue] (A) circle (2pt) node [left] {};
			\draw [fill=blue] (B) circle (2pt) node [right] {};
			\draw [fill=blue] (C) circle (2pt) node [right] {};
			\draw [fill=blue] (D) circle (2pt) node [above] {};
			\draw [-latex, red, thick] (A) -- (B);
			\draw [-latex, red, thick] (A) -- (C);
            \draw [-latex, red, thick] (A) -- (D);
			
			\draw [fill=blue] (A) circle (2pt) node [left] {};
			\draw [fill=blue] (B) circle (2pt) node [right] {\( v_1 \)};
			\draw [fill=blue] (C) circle (2pt) node [right] {\( v_2 \)};
			\draw [fill=blue] (D) circle (2pt) node [above] {\( v_3 \)};
     		\draw [fill=blue] (E) circle (2pt) node [above] {};
     		\draw [fill=blue] (F) circle (2pt) node [above] {};
            
			\draw [-latex, red, thick] (A) -- (B);
			\draw [-latex, red, thick] (A) -- (C);
            \draw [-latex, red, thick] (A) -- (D);
            \draw [-latex, ForestGreen, thick] (A) -- (E);
            \draw [-latex, Blue, thick] (A) -- (F);
			\end{tikzpicture}
			}\qquad\qquad
			\subfloat[]{
			\begin{tikzpicture}
			\coordinate (A) at (0,0);{}
			\coordinate (B) at (1,-0.2);
			\coordinate (C) at (1.2,0.3);
            \coordinate (D) at (1,0.5);
            \coordinate (E) at (1,0);
            \coordinate (F) at (-0.707, 0.707);

			\draw [fill=blue] (A) circle (2pt) node [left] {};
			\draw [fill=blue] (A) circle (2pt) node [left] {};
			\draw [fill=blue] (B) circle (2pt) node [right] {\( v_1 \)};
			\draw [fill=blue] (C) circle (2pt) node [right] {\( v_2 \)};
     		\draw [fill=blue] (D) circle (2pt) node [above] {\( v_3 \)};
     		\draw [fill=blue] (E) circle (2pt) node [above] {};
     		\draw [fill=blue] (F) circle (2pt) node [above] {};
            
			\draw [-latex, red, thick] (A) -- (B);
			\draw [-latex, red, thick] (A) -- (C);
            \draw [-latex, red, thick] (A) -- (D);
            \draw [-latex, ForestGreen, thick] (A) -- (E);
            \draw [-latex, Blue, thick] (A) -- (F);
			\end{tikzpicture}
			}
			\caption{Examples of Frames in $\R{2}$. The red vectors in both figures constitute frames of $\R{2}$. The frame constituted by the red vectors in figure (A) is tight.}
\label{fig:frames-example}
\end{figure}
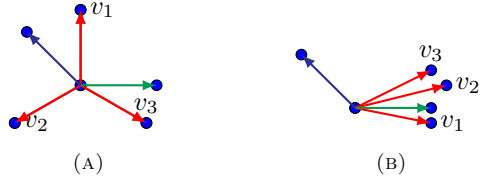

Intuitively, a frame is tight when the collection of vectors constituting the frame are as spread out in space as possible. For instance, consider the two sets of red vectors in $\R{2}$, each containing three elements as shown in Figure \ref{fig:frames-example}. Clearly, both these sets span $\R{2}$ and therefore constitute frames of $\R{2}$. However, the red-coloured vectors in Figure \ref{fig:frames-example} (A) are much more spread out than those in Figure \ref{fig:frames-example} (B). In Figure \ref{fig:frames-example} (B) the green-coloured vector is more or less aligned with all the three vectors of the frame and a large value of $C$ is obtained if the sum given in Definition \ref{d: frame} is computed with $v$ as the green-coloured vector. On the other hand, the blue-coloured vector in Figure \ref{fig:frames-example} (B) is almost orthogonal to all vectors of the frame, which leads to a small value of $c$ when the same computation is done with $v$ as the blue-coloured vector, thus resulting in a large gap between the values of $c$ and $C$. The same green and blue-coloured vectors shown in Figure \ref{fig:frames-example} (A) however are more or less equally aligned with respect to the frame vectors, and so one expects the gap in $c$ and $C$ to be much lesser. In fact, the red-coloured frame in Figure \ref{fig:frames-example} (A) is tight.

In the subsequent discussions in this article, we will almost exclusively deal with countably infinite frames and henceforth will use $\N$ as the index set instead of $\indexset$.

\begin{definition}
	\label{d:frame operators}
	For a given frame $\ftf$,
	\begin{itemize}[label=\(\circ\), leftmargin=*]
	\item the \textit{analysis operator} is defined by
		\[
			H_n \ni v\mapsto T_{ \ftf } (v) \coloneqq \big(\inprod{v_i}{v}_{H_n}\big)_{i \in \N} \in \ltwo{\R{}},
		\]
	\item the \textit{synthesis operator} (which is the adjoint of the analysis operator) is defined by
		\[
			\ltwo{\R{}} \ni \ltwoelement \mapsto T_{ \ftf }^{*} (\ltwoelement) \coloneqq \sum_{i=1}^{+\infty} \ltwoelement_i v_i \in H_n,
		\]
	\item the \textit{frame operator} is defined by
		\[
			H_n\ni v\mapsto G_{\ftf} (v) \coloneqq \sum_{i=1}^{+\infty} \inprod{v_i}{v}_{H_n}v_i \in H_n,
		\]
		which is the composition of the synthesis and the analysis operators (in that order).
\end{itemize}
\end{definition}
 It can be shown that whenever a sequence \( \ftf \) of vectors constitute a frame, the constants \( c \) and \( C \) are the smallest and the largest eigenvalues respectively of the corresponding frame operator. Thus, we see that whenever a frame is tight, the smallest and the largest eigenvalues of the corresponding frame operator are equal, which implies that all the eigenvalues are identical. Since the frame operator is self adjoint, we conclude that a frame is tight if and only if its frame operator is an appropriate multiple of the identity operator.

With this much background on frames, we recall that set of admissible controls in \eqref{eq:lti-system} is
\[
	\controlprofile = \aset[\bigg]{ u:[0, \horizon]\lra\R{m} \suchthat \int_{0}^{\horizon} \norm{u(t)}^2 \,\dd t < +\infty };
\]
it is a vector space, and is equipped with the natural inner product
\[
	\controlprofile \times \controlprofile \ni (u_1, u_2) \longmapsto \inprod{u_1}{u_2}_{\controlprofile} \Let \int_0^\horizon \inprod{u_1(t)}{u_2(t)} \,\dd t,
\]
with respect to which \(\controlprofile\) becomes a separable Hilbert space, for which the following assertions are classical:
\begin{proposition}[{\cite[Theorem 7.18, p.\ 139]{ref:Cla-13}}]
	\label{p: sep hilbert space}
	If \( \orthobasis \) is a countable orthonormal basis of the Hilbert space \( \controlprofile \), then for any $u \in \controlprofile$ we have
	\begin{equation*}
		u = \sum_{i=1}^{+\infty} \inprod{\orthobasiselement{i}}{u}_{\controlprofile} \orthobasiselement{i},
	\end{equation*}
	where the convergence is understood in the sense of the norm, and
	\[
		\norm{u}_{\controlprofile}^2 \Let \inprod{u}{u}_{\controlprofile} = \sum_{i = 1}^{+\infty} \abs[\big]{\inprod{\orthobasiselement{i}}{u}_{\controlprofile}}^2.
	\]
	Conversely, given any sequence $\alpha \Let (\alpha_i)_{i\in\N} \in \ltwo{\R{}}$, the series $\sum_{i = 1}^{+\infty} \alpha_i \orthobasiselement{i} $ converges to some element $u \in \controlprofile$ such that $\alpha_i = \inprod{\orthobasiselement{i}}{u}_{\controlprofile}$ for each \(i\) and $\norm{u}_{\controlprofile}^2 = \sum_{i = 1}^{+\infty} \alpha_i^2$. 
\end{proposition}

For the linear system \eqref{eq:lti-system} we define the \emph{end-point mapping} at time $\horizon$ by
\begin{equation}\label{end point map}
	\controlprofile \ni u \longmapsto \EPmap{\horizon}(u) \coloneqq \int_{0}^{\horizon} \exp{(\horizon - t)A}Bu(t)\, \dd t\in\R{\sysdim}.
\end{equation} 
It is well known that the image under the end-point map of any element $u \in \controlprofile$ is precisely the state to which the LTI system \eqref{eq:lti-system} is transferred at time $\horizon$ by the control signal $t\mapsto u(t)$ when initialized at the origin of the state space \( \R{\sysdim} \) at time $0$. The map \(\EPmap{\horizon}\) is, clearly, a continuous linear map from the Hilbert space $\controlprofile$ into $\R{\sysdim}$. We now establish that for a countable orthonormal basis \(\orthobasis\) of \(\controlprofile\), the sequence of vectors \( \big( \EPmap{\horizon} ( \orthobasiselement{i} ) \big)_{i \in \mathbb{N} } \) constitutes a \emph{frame} of \( \R{n} \) if and only if the LTI system \eqref{eq:lti-system} is controllable:

\begin{theorem}
	\label{th: lti frame}
	Let \( \orthobasis \) be any orthonormal basis of \( \controlprofile \), and let us define
	\begin{equation}
		\label{eq: lti frame}
		\ltiframeelement{i} \Let \EPmap{\horizon}(\orthobasiselement{i}) \quad \text{for } i \in \N.
	\end{equation}
	If \( u \in \controlprofile \), then
	\begin{equation}
		\label{l2 ep map}
		\EPmap{\horizon}(u) = \sum_{i=1}^{+\infty} \inprod{\orthobasiselement{i}}{u}_{\controlprofile} \ltiframeelement{i}.
	\end{equation}
	Moreover, the sequence \( (\ltiframeelement{i})_{i \in \N} \subset\R{\sysdim}\) satisfies the following properties:
		\begin{enumerate}[label=\textup{(\roman*)}, leftmargin=*, widest=iii, align=right]
			\item \label{th: lti frame:1} \(\sum_{i=1}^{+\infty} \norm{v_i}^2 < +\infty\),
			\item \label{th: lti frame:2} there exists \(C > 0\) such that \(\sum_{i=1}^{+\infty} \abs{\inprod{\ltiframeelement{i}}{v}}^2 \leq C\|v\|^2\) for all \(v \in \R{\sysdim}\), and
			\item \label{th: lti frame:3} the linear operator $G_{\ftf}$ (as defined in Definition \ref{d:frame operators}) is continuous.
	\end{enumerate}
	Furthermore, $(\ltiframeelement{i})_{i \in \N}$ is a frame of $\R{\sysdim}$ if and only if the LTI system \eqref{eq:lti-system} is controllable.
\end{theorem}

\begin{proof}
	We observe that since $\EPmap{\horizon}: \controlprofile \lra \R{\sysdim}$ is a continuous linear map, the property \eqref{l2 ep map} follows at once.

	We start with \ref{th: lti frame:1}. Continuity of the end point map \( \EPmap{\horizon} \) ensures the existence of a well-defined linear adjoint map $\adjoint{\EPmap{\horizon}}: \R{\sysdim} \lra \controlprofile$ that satisfies
	\begin{equation*}
		\inprod{\adjoint{\EPmap{\horizon}}(v)}{u}_{\controlprofile} = \inprod{v}{\EPmap{\horizon}(u)} \quad \text{ for all } v \in \R{\sysdim} \text{ and } u \in \controlprofile.
	\end{equation*}
	We immediately observe that
	\begin{align*}
		\trace(\adjoint{\EPmap{\horizon}} \EPmap{\horizon}) &= \sum_{i=1}^{+\infty} \inprod{\orthobasiselement{i}}{\adjoint{\EPmap{\horizon}} \EPmap{\horizon}(\orthobasiselement{i})}_{\controlprofile} \\
		&= \sum_{i=1}^{+\infty} \inprod{\EPmap{\horizon}(\orthobasiselement{i})} {\EPmap{\horizon}(\orthobasiselement{i})} \\
		&= \sum_{i=1}^{+\infty} \|\ltiframeelement{i}\|^2,
	\end{align*}
	and moreover it is easily verified that:
	\[
		\trace(\adjoint{\EPmap{\horizon}} \EPmap{\horizon}) = \trace(\EPmap{\horizon} \adjoint{\EPmap{\horizon}}) .
	\]
	However, since $\EPmap{\horizon} \adjoint{\EPmap{\horizon}}$ is a linear transformation on a finite dimensional Hilbert space, we obtain the following
	\[
	\sum_{i=1}^{+\infty} \norm{\ltiframeelement{i}}^2 = \trace(\adjoint{\EPmap{\horizon}} \EPmap{\horizon}) = \trace(\EPmap{\horizon} \adjoint{\EPmap{\horizon}}) < \infty.
	\] 
This proves \ref{th: lti frame:1}.

	\ref{th: lti frame:2} follows immediately from the Cauchy-Bunyakovsky-Schwartz inequality:
	\[
		\sum_{i=1}^{+\infty} \abs{\inprod{\ltiframeelement{i}}{v}}^2 \leq \biggl( \sum_{i=1}^{+\infty} \|\ltiframeelement{i}\|^2 \biggr) \|v\|^2.
	\]

	We establish \ref{th: lti frame:3} by showing that for each $v \in \R{\sysdim}$ the infinite sum in the definition of $G_{(\ltiframeelement{i})_{i\in\N}}$ is absolutely summable and hence convergent to a well-defined limit in $\R{\sysdim}$. Indeed, by continuity of the norm and the Cauchy-Bunyakovsky-Schwarz inequality,
	\begin{align*}
		\norm{ \sum_{i=1}^{+\infty} \inprod{\ltiframeelement{i}}{v}\ltiframeelement{i} } &\leq \sum_{i=1}^{+\infty} \norm{v} \norm{\ltiframeelement{i}}^2 = \norm{ v } \biggl( \sum_{i=1}^{+\infty}\norm{ \ltiframeelement{i} }^2 \biggr),
	\end{align*}
	which shows continuity of the map \(G_{(\ltiframeelement{i})_{i\in\N}}\).

	We prove the final statement by recalling our earlier observation that a sequence of vectors is a frame of $\R{\sysdim}$ if and only if the sequence is square summable and they span $\R{\sysdim}$. We conclude that the LTI system \eqref{eq:lti-system} is controllable if and only if $\Span(\ltiframeelement{i})_ {i \in \N} = \R{\sysdim}$ as asserted.
\end{proof}

\section{Optimal Control and the Classical MOQs}
\label{sec: optimal control}

We define the \emph{control effort} $J: \controlprofile \lra \R{}$ of a control \( u \in \controlprofile \) to be
\[
	J(u) \coloneqq \int_{0}^{\horizon} \inprod{u(t)}{u(t)} \, \dd t.
\]
This particular control effort is of great practical relevance since it is the energy required to drive the system with the control \(t\mapsto u(t)\). It is, therefore, quite natural to minimize this control effort. For a given \( x \in \R{n} \) we consider the following optimal control problem:
\begin{equation}
\label{eq:optimal-control-problem}
	\left\{
	\begin{aligned}
		& \minimize_{u \in \controlprofile} &&  J(u)\\
		&\sbjto && \EPmap{T}(u) = x.
	\end{aligned}
	\right.
\end{equation}
The optimal control problem \eqref{eq:optimal-control-problem} is well-studied in classical optimal control theory, and its solution can be described analytically in terms of the \emph{controllability Gramian} at time $\horizon$ defined by
\begin{equation}
	\label{eq:control-grammian}
	\congrammian \coloneqq \int_{0}^{\horizon} \exp{tA}BB\transp\exp{t A\transp} \,\dd t. 
\end{equation}
Classical results \cite[Chapter 22]{ref:Cla-13} guarantee that there exists a unique solution of \eqref{eq:optimal-control-problem} under controllability of \eqref{eq:lti-system}, and the optimal control \( u^x \) and the optimal cost \( J^x \) that solve \eqref{eq:frame-optimal-control-problem} are given by
\begin{equation}
\label{eq:optimal-control-solution}
\begin{aligned}
	u^x (t) & = B\transp \exp{(\horizon - t)A\transp} \congrammian^{-1}x  \quad\text{for \(t\in[0, \horizon]\), and}\\
	J^x  & =  \inprod{x}{\congrammian^{-1} x}.
\end{aligned}
\end{equation}

\begin{remark}
	The similarity in our notation between the controllability Gramian and the frame operator is intentional, the motivation for which we shall see clearly in Proposition \ref{th: grammian-frame-eq} below.
\end{remark}

From Proposition \ref{p: sep hilbert space} we recall that every control \( u \in \controlprofile \) is uniquely determined a sequence in \( \ltwo{\R{}} \). Therefore, in view of \eqref{l2 ep map}, we recast the optimal control problem \eqref{eq:optimal-control-problem} completely in terms of the sequences in \( \ltwo{\R{}} \) that describe the respective control functions.
\begin{equation}
	\label{eq:frame-optimal-control-problem}
	\left\{
	\begin{aligned}
		& \minimize_{\ltwoelement \in \ltwo{\R{}}} && \sum_{i=1}^{+\infty} \ltwoelement_i^2\\
		& \sbjto && \sum_{i=1}^{+\infty} \ltwoelement_i \ltiframeelement{i} = x,
	\end{aligned}
	\right.
\end{equation}
which is equivalent to
\begin{equation}
	\label{eq:frame-optimal-control-problem-2}
	\left\{
	\begin{aligned}
		& \minimize_{\ltwoelement \in \ltwo{\R{}}} && \inprod{\ltwoelement}{\ltwoelement}_{\ltwo{\R{}}} \\
		& \sbjto && T_{\ftf}^{*}(\ltwoelement) = x.
	\end{aligned}
	\right.
\end{equation}
The equivalence between the problems \eqref{eq:optimal-control-problem} and \eqref{eq:frame-optimal-control-problem-2} gives us the following result:

\begin{proposition}
	\label{th: grammian-frame-eq}
	If $(\ltiframeelement{i})_{i \in \N}$ is the frame in $\R{\sysdim}$ generated according to \eqref{eq: lti frame} by the LTI system \eqref{eq:lti-system} at time $\horizon$, then
	\begin{equation}
		\label{eq:grammian-frame-operator-equivalence}
		G_{(\ltiframeelement{i})_{i \in \N}} = G_{(A, B,\horizon)}.
	\end{equation}
\end{proposition}
\begin{proof}
	Since both $G_{(\ltiframeelement{i})_{i \in \N}}$ and $G_{(A, B,\horizon)}$ are symmetric positive definite matrices, it is enough to show that 
	\begin{equation*}
		\inprod{x}{G_{\ftf}^{-1} x} = \inprod{x}{G_{(A, B,\horizon)}^{-1} x} \quad \text{for all } x \in \R{\sysdim},
	\end{equation*}
	for then a standard argument involving the polarization identity suffices to conclude that
	\[
		\inprod{x}{G_{\ftf}^{-1} y} = \inprod{x}{G_{(A, B,\horizon)}^{-1} y} \quad \text{for all } x, y \in \R{\sysdim},
	\]
	which in turn implies the assertion immediately.\footnote{Recall that the polarization identity states that \(\inprod{x}{Ay} + \inprod{y}{Ax} = \frac{1}{2}\bigl( \inprod{x+y}{A(x+y)} + \inprod{x-y}{A(x-y)} \bigr)\) for all \(x, y\in\R{\nu}\) and any matrix \(A\in\R{\nu\times\nu}\). If \(A\) is symmetric, then \(\inprod{x}{Ay} = \inprod{y}{Ax}\), which shows that \(\inprod{x}{Ay} = 0\) whenever \(\inprod{x+y}{A(x+y)} = \inprod{x-y}{A(x-y)} = 0\). It follows at once that if \(A, B\in\R{\nu\times\nu}\) are two symmetric matrices and \(\inprod{x}{Ax} = \inprod{x}{Bx}\) for \(x\in\R{\nu}\), then \(A = B\).%
	}
	To this end, we already know that
	\begin{equation*}
		\text{for every \( x \in \R{n} \)}\quad J^x = \inprod{x}{G_{(A, B,\horizon)}^{-1} x},
	\end{equation*}
	so that it suffices now to verify that
	\begin{equation*}
		J^x = \inprod{x}{G_{\ftf}^{-1} x}.
	\end{equation*}
	Since the problems \eqref{eq:optimal-control-problem} and \eqref{eq:frame-optimal-control-problem-2} are equivalent, we conclude that $ J^x $ is the optimum value achieved in \eqref{eq:frame-optimal-control-problem-2} as well. It is a classical result (see, e.g., \cite[Section 6.11]{luenopti}) that \( \inprod{x}{\bigl(T_{\ftf}T^{*}_{\ftf}\bigr)^{-1}x} \) is the optimum value achieved in the problem \eqref{eq:frame-optimal-control-problem-2}. In view of the fact that $\bigl(T_{\ftf}T^*_{\ftf}\bigr) = G_{\ftf}$, we conclude that \eqref{eq:grammian-frame-operator-equivalence} holds.
\end{proof}

\begin{remark}
	The equality in \eqref{eq:grammian-frame-operator-equivalence} is true even if the sequence of vectors \( \ftf \) is not a frame of \( \R{n} \). Let \( x \in \R{n} \) be arbitrary, and define \( [0,1] \ni t \longmapsto z(t) \coloneqq \big( B\transp \exp{(\horizon - t)A\transp} x \big) \). Clearly, \( z \in \controlprofile \). We observe that
\begin{align*}
    \congrammian x &= \int_0^\horizon \exp{(\horizon - t)A} B \big( B\transp \exp{(\horizon - t)A\transp} x \big) \,\dd t\\
	& = \EPmap{\horizon} (z) = \sum_{i = 0}^{+\infty} \inprod{\orthobasiselement{i}}{z}_{\controlprofile} \EPmap{\horizon}(\orthobasiselement{i}) \\
	& = \sum_{i = 1}^{+\infty} \inprod{\orthobasiselement{i}}{z}_{\controlprofile} \ltiframeelement{i},
\end{align*}
and by definition we have
\begin{align*}
	\inprod{\orthobasiselement{i}}{z}_{\controlprofile} & = \int_0^\horizon \big( B\transp \exp{(\horizon - t)A\transp} x \big)\transp \orthobasiselement{i}(t) \,\dd t  = \inprod{x}{\EPmap{\horizon} (\orthobasiselement{i})} = \inprod{x}{\ltiframeelement{i}} = \inprod{\ltiframeelement{i}}{x}.
\end{align*}
Collecting the equalities above we get
\[
	\congrammian x = \sum_{i = 0}^{+\infty} \inprod{\ltiframeelement{i}}{x} \ltiframeelement{i} = G_{\ftf} x.
\]
\end{remark}

The equality in \eqref{eq:optimal-control-solution} tells us that the minimum control effort required to transfer the LTI system \eqref{eq:lti-system} is completely determined by its controllability Gramian. The following three MOQs proposed in \cite{ref:MulWeb-72} are based on this fact:
\begin{enumerate}[label=\textup{(\roman*)}, leftmargin=*, widest=iii, align=right]
	\item \( \trace \big( \congrammian^{-1} \big) \): This quantity is proportional to the average optimal control effort needed to transfer the system state from origin (i.e., \( x_0 = 0\)) to a random point that is uniformly distributed on the unit sphere. In fact, 
	\begin{align*}
		\frac{1}{\sysdim} \trace \big( \congrammian^{-1} \big) = \frac{\int_{\|x\| = 1} \inprod{x}{\congrammian^{-1} x} \dd x}{\int_{\|x\| = 1}\inprod{x}{x} \dd x}.
	\end{align*}
	\item \( \lambda_{\min}^{-1} \big( \congrammian \big) \): This quantity gives the maximum control effort needed to transfer the system state from origin to any point on the unit sphere, which easily follows from the fact that
	\begin{align*}
		\lambda_{\min}^{-1} \big( \congrammian \big) = \lambda_{\max} \big( \congrammian^{-1} \big) = \max_{\|x\|=1} \inprod{x}{\congrammian^{-1} x}.
	\end{align*}
	\item \( \det \big( \congrammian \big) \): This quantity is proportional to the volume of the ellipsoid containing points to which the system state can be transferred to from the origin using at most unit control effort. Indeed,
	\begin{align*}
		\text{Vol} \Big( \aset{ x \in \R{\sysdim} \suchthat \inprod{x}{\congrammian^{-1} x} \leq 1} \Big) \ \propto \ \sqrt{\det(\congrammian)}.
	\end{align*}
\end{enumerate}

The LTI system \eqref{eq:lti-system} is controllable if and only if the Gramian $\congrammian$ is invertible. So $\trace \big( \congrammian^{-1} \big)$ is well-defined only when the system is controllable. Similarly, $\lambda_{\min}^{-1} \big( \congrammian \big)$ attains a finite value and $\det \big( \congrammian \big)$ has a non-zero value if and only if the system is  controllable. Therefore all the three quantities defined above give significant information regarding the minimum energy state transfer problem and can be used to clearly distinguish between controllable and uncontrollable systems. So it is reasonable to say that all three of the values given above are valid measures of quality of the LTI system \eqref{eq:lti-system}. In view of this fact, from now on we will refer to the above three MOQs as the classical MOQs. Since all properties of the LTI system \eqref{eq:lti-system} that we are interested in this article are determined completely by the frame $(\ltiframeelement{i})_{i \in \N}$, including the definition of the classical MOQs, we will use the terms MOQ of the LTI system \eqref{eq:lti-system} and MOQ of the frame generated by the LTI system interchangeably. We mention that the three classical MOQs are generally not correlated in any way; one could increase any one of them arbitrarily while keeping the value of the other fixed. We would also like to point out that to define the classical MOQ (ii), unlike (i) and (iii), we do not require any notion of a volume on the state space. Hence it can be readily extended with some minor modifications to the case of infinite dimensional linear systems. We refer the reader to \cite{ref:Lio-94} for a discussion on such an extension. 

\section{Optimization of the classical MOQs}
\label{sec: optimization}

In this section we find what the orientation of the frame vectors generated by an 
LTI system should be so that the LTI system is best in terms of each of the 
three classical MOQs mentioned above. In other words, we minimize the three classical MOQs with
respect to the frame vectors subject to the constraint that the length of each
of the frame vector is kept fixed. 

Let \( ( \alpha_i )_{i \in \N} \) be a sequence satisfying
\begin{equation}
	\label{e:key}
	\begin{cases}
		\alpha_i > 0\text{ for each }i,\\
		(\alpha_i)_{i\in\N} \text{ is non-increasing}, \sum_{i = 1}^{+\infty} \alpha_i < +\infty,\text{ and}\\
		\alpha_1 \leq \frac{1}{n} \sum_{i = 1}^{+\infty} \alpha_i \teL \A.
	\end{cases}
\end{equation}
In order to find the optimal orientation of the vectors of the frame, we optimize the three objective functions  \( \trace \big( G_{ \ftf }^{-1} \big) \), \( \lambda_{\min}^{-1} \big( G_{ \ftf } \big) \) and \( \det \big( G_{ \ftf } \big) \), subject to the constraint that the lengths of the vectors are fixed, i.e., \( \langle v_i, v_i \rangle = \alpha_i \) for all \( i \in \N \). Thus, we have the following three optimization problems:
\begin{equation}
\label{eq:trace-MOQ-optimization}
\begin{aligned}
\begin{cases}
   \underset{ ( v_i )_{i \in \N} }{\text{minimize}}
   & \trace ( G_{ ( v_1, \ldots, v_K ) }^{-1} )  \\
   \text{subject to}
   & \langle v_i, v_i \rangle = \alpha_i \quad \text{for all \( i \in \N \)}, \\
   & \Span\ftf = \R{n},
\end{cases}
\end{aligned}
\end{equation}

\begin{equation}
\label{eq:min-eigenvalue-MOQ-optimization}
\begin{aligned}
\begin{cases}
   \underset{ ( v_i )_{i \in \N }} {\text{minimize}}
 & \lambda_{\min}^{-1} ( G_{ \ftf } )  \\
   \text{subject to}
 & \langle v_i, v_i \rangle = \alpha_i\quad \text{for all \( i \in \N \)}, \\
 & \Span\ftf = \R{n},
\end{cases}
\end{aligned}
\end{equation}
and
\begin{equation}
\label{eq:det-MOQ-optimization}
\begin{aligned}
\begin{cases}
   \underset{ ( v_i )_{i \in \N} }{\text{maximize}}
 & \det( G_{ \ftf })  \\
   \text{subject to}
 & \langle v_i, v_i \rangle = \alpha_i\quad \text{for all \( i \in \N \)}, \\
 & \Span\ftf = \R{n}.
\end{cases}
\end{aligned}
\end{equation}

Surprisingly, all the three problems \eqref{eq:trace-MOQ-optimization}, \eqref{eq:min-eigenvalue-MOQ-optimization} and \eqref{eq:det-MOQ-optimization} have unique optimizers and they coincide, which is the content of the following theorem:

\begin{theorem}
	\label{th: frames-optimality}
	The optimal values of the three optimization problems \eqref{eq:trace-MOQ-optimization}, \eqref{eq:min-eigenvalue-MOQ-optimization} and \eqref{eq:det-MOQ-optimization} are attained when $G_{\ftf} = G\opt = \A I_n$, where $\A \Let \frac{1}{\sysdim} \sum_{i = 1}^{+\infty} \alpha_i$. Moreover, \(G\opt\) is the unique optimizer of each of these problems.
\end{theorem}

A proof of Theorem \ref{th: frames-optimality} is given at the end of this section. In all the above three problems, even though the optimization is carried out over the frame vectors $(\ltiframeelement{i})_{i \in \N}$, the objective function in each of these problems depend only on the frame operator and not on the vectors themselves. We will observe in the discussion below that the constraints can also be recast completely in terms of the frame operator. We start by recalling the following definition of majorization:

\begin{definition}
	\label{d:majorization}
	Consider a non-increasing finite sequence \( \eigval{} \coloneqq (\eigval{1}, \ldots, \eigval{n}) \) and for \(K\in\N\cup\{+\infty\}\) a sequence \( \alpha = (\alpha_i)_{i=1}^K\), with positive real numbers as their entries. We define the relation \( \eigval{} \succ  \alpha \) if the following two conditions hold:
	\begin{equation}
		\label{eq:rod-majorization-inequality}
		\left\{
		\begin{aligned}
			\sum_{i = 1}^m \eigval{i} & \geq \sum_{i = 1}^m \alpha_i \quad \text{for all \( m = 1, \ldots, n - 1 \), and} \\
			\sum_{i = 1}^n \eigval{i} & = \sum_{i = 1}^{K} \alpha_i.
		\end{aligned}
		\right.
	\end{equation}
\end{definition}
The conditions in \eqref{eq:rod-majorization-inequality} are analogues of the standard majorization conditions \cite[Chapter 1]{ref:MarOlkArn-11}. The following version of the Schur-Horn theorem establishes a connection between the preceding majorization relation and frame operators corresponding to given frames.

\begin{lemma}\cite[Theorem 4.7]{antezana2007schur}
	\label{lemma:rod}
	For any given sequence \( \alpha = (\alpha_i)_{i\in\N} \) of positive real numbers satisfying \eqref{e:key} and a symmetric and non-negative definite matrix \( G \in \posdef{n} \) with \( \N \cup \{+\infty\} \ni K \geq n \), the following statements are equivalent:
	\begin{itemize}[label=\(\circ\), leftmargin=*]
		\item There exists a sequence of vectors \( \ftf \subset \R{n}\) such that \( G = G_{\ftf} \) and \( \langle v_i, v_i \rangle = \alpha_i \) for all \( i = 1, 2, \ldots, K \).
		\item \( \eigval{}(G) \succ \alpha \).
	\end{itemize}
\end{lemma}

Lemma \ref{lemma:rod} states that the mapping
\begin{equation}
	\label{eq:rod-equivalence-map}
	\ltwo{\R{\sysdim}} \ni \ftf \longmapsto G_{ \ftf } \in \possemdef{\sysdim}, 
\end{equation}
establishes a correspondence between the sequence of vectors \( \ftf \) that are feasible for the optimization problems \eqref{eq:trace-MOQ-optimization}, \eqref{eq:min-eigenvalue-MOQ-optimization} and \eqref{eq:det-MOQ-optimization}, and the set of positive definite matrices $ G\in\possemdef{\sysdim}$ such that $\eigval{}(G) \succ \alpha$. The objective functions in each of those problems are also clearly invariant under the mapping \eqref{eq:rod-equivalence-map}. Consequently, we can recast the problems \eqref{eq:trace-MOQ-optimization}, \eqref{eq:min-eigenvalue-MOQ-optimization}, \eqref{eq:det-MOQ-optimization} into, respectively, the following equivalent problems over positive definite matrices:
\begin{equation}
\label{eq:trace-MOQ-optimization-1}
\begin{dcases}
	\minimize_{G \; \in \; \posdef{n}} & \trace ( G^{-1} )\\
	\sbjto & \eigval{}(G) \succ \alpha,
\end{dcases}
\end{equation}
\begin{equation}
\label{eq:min-eigenvalue-MOQ-optimization-1}
\begin{dcases}
	\minimize_{G \; \in \; \posdef{n}} & \lambda_{\min}(G)\inverse\\
	\sbjto & \eigval{}(G) \succ \alpha ,
\end{dcases}
\end{equation}
\begin{equation}
\label{eq:det-MOQ-optimization-1}
\begin{dcases}
	\minimize_{G \; \in \; \posdef{n}} & \det (G)\\
	\sbjto & \eigval{}(G) \succ \alpha.
\end{dcases}
\end{equation}

We now demonstrate one by one that the unique optimizer of the problems \eqref{eq:trace-MOQ-optimization-1}, \eqref{eq:min-eigenvalue-MOQ-optimization-1}, and \eqref{eq:det-MOQ-optimization-1} is \( G\opt \coloneqq \A I_n \). 

\begin{lemma}
\label{lemma:trace-moq}
	Given a sequence $\alpha = (\alpha_i)_{i\in\N}$ of real numbers satisfying \eqref{e:key}, $G\opt \coloneqq \A I_n$ is the unique optimizer of the optimization problem \eqref{eq:trace-MOQ-optimization-1}, where \(\A\) is the constant defined in Theorem \ref{th: frames-optimality}.
\end{lemma}
\begin{proof}
	Let us consider the following optimization problem:
	\begin{equation}
		\label{eq:trace-MOQ-optimization-different}
		\left\{
		\begin{aligned}
			& \minimize_{G \; \in \; \posdef{n}} && \trace ( G\inverse )\\
			& \sbjto && \trace(G) = n \A.
		\end{aligned}
		\right.
	\end{equation}
	The optimization problem \eqref{eq:trace-MOQ-optimization-different} is the same as the problem in \cite[(22), p.\ 17]{ref:SheCha-17} with \( \Sigma_V = I_n \), and its unique solution is given in \cite[(25)]{ref:SheCha-17}. Therefore, from \cite{ref:SheCha-17} we conclude that \( G\opt \) is the unique optimal solution to the problem \eqref{eq:trace-MOQ-optimization-different}.
	
	It remains to establish an equivalence between \eqref{eq:trace-MOQ-optimization-different} and \eqref{eq:trace-MOQ-optimization-1}. From \eqref{eq:rod-majorization-inequality} we know that if $\eigval{}(G) \succ \alpha$, then 
	\begin{equation*}
		\trace(G) = \sum_{i = 1}^n \eigval{i} = \sum_{i = 1}^{+\infty} \alpha_i.
	\end{equation*}
	Therefore, the optimum value of \eqref{eq:trace-MOQ-optimization-different} is at most equal to the optimum value of \eqref{eq:trace-MOQ-optimization-1} (if it exists) since \eqref{eq:trace-MOQ-optimization-different} involves minimization of the trace MOQ over a larger set of positive definite matrices. However, observe that $\eigval{}(G\opt) \succ \alpha$, since
	\begin{align*}
		\sum_{i = 1}^m \eigval{i} &= m \A = \frac{m}{\sysdim} \sum_{i=1}^{+\infty} \alpha_{i} \geq \sum_{i = 1}^m \alpha_i, \text{ and}\\
		\sum_{i = 1}^n \eigval{i} &= n \A = \sum_{i=1}^{+\infty} \alpha_{i},
	\end{align*}
	which shows that $G\opt$ is also feasible for \eqref{eq:trace-MOQ-optimization-1}. Together with our earlier observation that $\trace(G\opt)$ is at most equal to the optimum value of \eqref{eq:trace-MOQ-optimization-1}, this implies that $G\opt$ is an optimizer of \eqref{eq:trace-MOQ-optimization-1}. Since $G\opt$ is the unique optimizer of $\eqref{eq:trace-MOQ-optimization-different}$, it is also the unique optimizer of \eqref{eq:trace-MOQ-optimization-1}, completing our proof.
\end{proof}

\begin{lemma}
\label{lemma:min-eigenvalue-moq}
	Given a sequence $\alpha = (\alpha_i)_{i\in\N}$ of real numbers satisfying \eqref{e:key}, $G\opt \coloneqq \A I_n$ is the unique optimizer of the optimization problem \eqref{eq:min-eigenvalue-MOQ-optimization-1}, where \(\A\) is the constant defined in Theorem \ref{th: frames-optimality}.
\end{lemma}
\begin{proof}
	Both the objective function and the feasible set in \eqref{eq:min-eigenvalue-MOQ-optimization-1} are determined completely by the set of eigenvalues $(\eigval{1}, \ldots, \eigval{\sysdim})$ of $G$. Therefore, we consider the following optimization problem that is equivalent to \eqref{eq:min-eigenvalue-MOQ-optimization-1}:
	\begin{equation}
		\label{eq:min-eigenvalue-MOQ-optimization-2}
		\left\{
		\begin{aligned}
			& \minimize_{\substack{\eigval{i} > 0\\i=1, \ldots, \sysdim}} && \big( \min \{ \eigval{1}, \ldots, \eigval{n} \} \big)^{-1}\\
			& \sbjto && \eigval{} \succ \alpha. 
		\end{aligned}
		\right.
	\end{equation}
	For every sequence $ ( \eigval{1}, \ldots, \eigval{n} ) $ that is feasible for the optimization problem \eqref{eq:min-eigenvalue-MOQ-optimization-2}, we see that  \( \sum_{i = 1}^n \eigval{i} = \sum_{i = 1}^{+\infty} \alpha_i = n \A \). Therefore,
	\[
		\min \{ \eigval{1}, \ldots, \eigval{n} \} \leq \A,
	\]
	which implies that
	\[
		\big( \min \{ \eigval{1}, \ldots, \eigval{n} \} \big)\inverse \geq { \A }\inverse.
	\]
	Therefore, the value \( { \A }\inverse \) is at most equal to the optimal value (if it exists) of the optimization problem \eqref{eq:min-eigenvalue-MOQ-optimization-2}. Let us define \( \eigval{i}\opt \coloneqq \A \) for all \( i = 1,2,\ldots,n \). We observed in the proof of Lemma \ref{lemma:trace-moq} that
	\begin{equation}
		\label{eq:min-eigval-solution-2}
		(\eigval{1}\opt, \ldots, \eigval{\sysdim}\opt) \succ \alpha.
	\end{equation}
	We also know that
	\begin{equation}
		\label{eq:min-eigval-solution-3}
		\big( \min \{ \eigval{1}\opt, \ldots, \eigval{n}\opt \} \big)^{-1} = { \A }^{-1}.
	\end{equation}
	This means that $(\eigval{1}\opt, \ldots, \eigval{\sysdim}\opt)$ is feasible for the optimization problem \eqref{eq:min-eigenvalue-MOQ-optimization-2}, and together with our earlier observation that $\big( \min \{ \eigval{1}\opt, \ldots, \eigval{n}\opt \} \big)^{-1}$ is at most equal to the optimal value of \eqref{eq:min-eigenvalue-MOQ-optimization-2}, this implies that $( \eigval{1}\opt, \ldots, \eigval{n}\opt )$ is an optimizer of \eqref{eq:min-eigenvalue-MOQ-optimization-2}. It is easily seen that $( \eigval{1}\opt, \ldots, \eigval{n}\opt )$ is the unique sequence of positive real numbers that satisfies both \eqref{eq:min-eigval-solution-2} and \eqref{eq:min-eigval-solution-3}. Therefore, $( \eigval{1}\opt, \ldots, \eigval{n}\opt )$ is the unique optimizer of \eqref{eq:min-eigenvalue-MOQ-optimization-2}.

	In view of the equivalence between the optimization problems \eqref{eq:min-eigenvalue-MOQ-optimization-1} and \eqref{eq:min-eigenvalue-MOQ-optimization-2}, we conclude that a matrix $G$ is an optimizer of \eqref{eq:min-eigenvalue-MOQ-optimization-1} if and only if
	\begin{equation}
		\label{eq:min-eigval-solution-1}
		\eigval{i}(G) = \A \quad \text{for all } i = 1, 2, \ldots, \sysdim.
	\end{equation}
	It is well-known that $G = \A I_n$ is the only positive definite matrix that satisfies \eqref{eq:min-eigval-solution-1}, and, consequently, $G\opt = \A I_n$ is the unique optimizer of \eqref{eq:min-eigenvalue-MOQ-optimization-1}.
\end{proof}

\begin{lemma}
\label{lemma:det-moq}
	Given a sequence $\alpha = (\alpha_i)_{i\in\N}$ of real numbers satisfying \eqref{e:key}, $G\opt \coloneqq \A I_n$ is the unique optimizer of the optimization problem \eqref{eq:det-MOQ-optimization-1}, where \(\A\) is the constant defined in Theorem \ref{th: frames-optimality}.
\end{lemma}
\begin{proof}
	Once again, the objective function in \eqref{eq:det-MOQ-optimization-1} depends only on the eigenvalues of $G$, and so the problem can be recast as the following equivalent one:
	\begin{equation}
		\label{eq:det-MOQ-optimization-2}
		\left\{
		\begin{aligned}
			& \minimize_{\substack{\eigval{i} > 0\\ i=1, \ldots, \sysdim}} && \prod_{i = 1}^{\sysdim} \eigval{i}\\
			& \sbjto && \eigval{} \succ \alpha.
		\end{aligned}
		\right.
	\end{equation}
	By the arithmetic mean -- geometric mean inequality we know that for any sequence $(\eigval{1}, \ldots, \eigval{\sysdim})$ of positive numbers,
	\[
		\prod_{i = 1}^n \eigval{i} \leq \Big( \frac{\sum_{i = 1}^n \eigval{i}}{\sysdim} \Big)^{\sysdim} \leq \A^{\sysdim}.
	\]
	Therefore, the value \( \A^{\sysdim} \) is at least equal to the optimal value of the optimization problem \eqref{eq:det-MOQ-optimization-2} (if it exists). Let us define \( \eigval{i}\opt \coloneqq \A \) for all \( i = 1,2,\ldots,n \). We observed in the proof of Lemma \ref{lemma:trace-moq} that
	\begin{equation}
		\label{eq:det-eigval-solution-2}
		(\eigval{1}\opt, \ldots, \eigval{\sysdim}\opt ) \succ \alpha.
	\end{equation}
	We also know that
	\begin{equation}
		\label{eq:det-eigval-solution-3}
		\prod_{i = 1}^n \eigval{i}\opt = \Big( \frac{\sum_{i = 1}^n \eigval{i}\opt}{\sysdim} \Big)^{\sysdim} = \A^{\sysdim}.
	\end{equation}
	This means that $(\eigval{1}\opt, \ldots, \eigval{\sysdim}\opt)$ is feasible for the optimization problem \eqref{eq:det-MOQ-optimization-2}. Together with our earlier observation that $\prod_{i = 1}^n \eigval{i}$ is at most equal to the optimal value of \eqref{eq:det-MOQ-optimization-2}, this implies that $( \eigval{1}\opt, \ldots, \eigval{n}\opt )$ is an optimizer of \eqref{eq:det-MOQ-optimization-2}. As remarked earlier, $( \eigval{1}\opt, \ldots, \eigval{n}\opt )$ is the unique sequence of positive real numbers that satisfies both \eqref{eq:det-eigval-solution-2} and \eqref{eq:det-eigval-solution-3}. Therefore, $( \eigval{1}\opt, \ldots, \eigval{n}\opt )$ is the unique optimizer of \eqref{eq:det-MOQ-optimization-2}.

	In view of the equivalence between the optimization problems \eqref{eq:det-MOQ-optimization-1} and \eqref{eq:det-MOQ-optimization-2}, we conclude that a matrix $G$ is an optimizer of \eqref{eq:min-eigenvalue-MOQ-optimization-1} if and only if
	\begin{equation}
		\label{eq:det-eigval-solution-1}
		\eigval{i}(G) = \A \quad \text{for all } i = 1, 2, \ldots, \sysdim.
	\end{equation}
	Since $G\opt = \A I_n$ is the only matrix that satisfies \eqref{eq:det-eigval-solution-1}, it is the unique optimizer of \eqref{eq:det-MOQ-optimization-1}.
\end{proof}

\begin{proof}[Proof of Theorem \ref{th: frames-optimality}]
	The assertion follows from Lemmas \ref{lemma:trace-moq}, \ref{lemma:min-eigenvalue-moq}, and \ref{lemma:det-moq}.
\end{proof}

\section{Frame-theoretic MOQ}
\label{sec: frame MOQ}
We have seen in the preceding section that all the three classical MOQs attain their optimal value precisely when the frame generated by the LTI system is tight. In addition, we discussed in the introduction how tight frames posses several additional desirable properties for representation of generic vectors. Thus, a measure of tightness of the frame $\ftf$, where the vectors \(v_i\) are defined by \eqref{eq: lti frame}, is a valid measure of quality of the LTI system \eqref{eq:lti-system}. 

\begin{definition}
	\label{d: moq}
	We define the quantity 
	\begin{equation*}
		\moq{A}{B}{T} \Let \frac{\trace{\bigl(\grammian{A}{B}{\horizon}}\bigr)}{\sqrt{\trace{\bigl(\grammian{A}{B}{\horizon}^2}\bigr)}},
	\end{equation*}
	which measures the extent of tightness of the frame generated by the LTI system \eqref{eq:lti-system}, as a measure of the quality of \eqref{eq:lti-system}. In the subsequent discussions we will refer to the MOQ defined in Definition \ref{d: moq} as the \emph{frame-theoretic MOQ}.
\end{definition}

Recall that for $A_1, A_2 \in \mat{\sysdim}{\sysdim}$, 
\begin{equation}
	\label{eq: inner prod}
	\inprod{A_1}{A_2}_{\mat{\sysdim}{\sysdim}} \Let \trace({A_1\transp A_2})
\end{equation}
defines an inner product on $\mat{\sysdim}{\sysdim}$ under which it is a Hilbert space. Since $\grammian{A}{B}{\horizon}$ is symmetric, we have
\begin{align*}
	\trace{\bigl(\grammian{A}{B}{\horizon}\bigr)} &= \trace{\bigl(I_{\sysdim}\transp \grammian{A}{B}{\horizon}\bigr)} = \inprod{\grammian{A}{B}{\horizon}}{I_{\sysdim}}_{\mat{\sysdim}{\sysdim}}, \text{ and}\\
	\trace{\bigl(\grammian{A}{B}{\horizon}^2\bigr)} &= \trace{\bigl(\grammian{A}{B}{\horizon}\transp \grammian{A}{B}{\horizon}\bigr)} = \inprod{\grammian{A}{B}{\horizon}}{\grammian{A}{B}{\horizon}}_{\mat{\sysdim}{\sysdim}}, 
\end{align*}
which in turn implies that
\begin{align}\label{eq: moq alt}
	\moq{A}{B}{\horizon} = \frac{\inprod{\grammian{A}{B}{\horizon}}{I_{\sysdim}}_{\mat{\sysdim}{\sysdim}}}{\sqrt{\inprod{\grammian{A}{B}{\horizon}}{\grammian{A}{B}{\horizon}}_{\mat{\sysdim}{\sysdim}}}}.
\end{align}

We had remarked in Definition \ref{d: moq} that the frame-theoretic MOQ \(\eta\) is a measure of tightness of the frame generated by the LTI system. From \eqref{eq: moq alt} we can see immediately that \(\eta\) is simply a multiple of the cosine of the angle that the Gramian $\grammian{A}{B}{\horizon}$ makes with the identity matrix $I_{\sysdim}$ in the Hilbert space $\mat{\sysdim}{\sysdim}$ equipped with the inner product \eqref{eq: inner prod}. Therefore, a higher value of the frame-theoretic MOQ \(\eta\) implies that $\grammian{A}{B}{\horizon}$ is more aligned with the identity matrix, which in turn means that the corresponding frame generated by the LTI system is tighter. In particular, this means the higher the value of the frame-theoretic MOQ, the better the LTI system is.

The following proposition provides some connections between the classical notion of controllability and certain properties of the frame-theoretic MOQ \(\eta\).

\begin{proposition}
	\label{p: moq properties}
	Consider the LTI system \eqref{eq:lti-system}. With $\moq{A}{B}{\horizon}$ being the frame-theoretic MOQ defined in Definition \ref{d: moq}, we have the following:
	\begin{enumerate}[label=\textup{(\roman*)}, leftmargin=*, widest=ii, align=right]
		\item \label{p:moq properties 1} For all $\horizon > 0$,
			\begin{equation}
				\label{eq:moq prop 1}
				\moq{A}{B}{\horizon} \leq \sqrt{\sysdim},
			\end{equation}
			and equality holds in \eqref{eq:moq prop 1} if and only if the frame \(\ftf\) generated by \eqref{eq:lti-system} at time $\horizon$ according to \eqref{eq: lti frame} is tight.
		\item \label{p:moq properties 2} If $\moq{A}{B}{\horizon} > \sqrt{d}$ for some positive integer \(d\le\sysdim\), then the dimension of the reachable subspace at time $\horizon$ is greater than $d$. In particular, if $\moq{A}{B}{\horizon} > \sqrt{\sysdim - 1}$, then \eqref{eq:lti-system} is controllable.
	\end{enumerate}
\end{proposition}
\begin{proof}
	In view of \eqref{eq: moq alt}, the inequality \eqref{eq:moq prop 1} in Proposition \ref{p: moq properties} is just a consequence of the Cauchy-Schwartz inequality. To see this, observe that
	\begin{align*} 
		\moq{A}{B}{\horizon} &= \frac{\inprod{\grammian{A}{B}{\horizon}}{I_{\sysdim}}}{\sqrt{\inprod{\grammian{A}{B}{\horizon}}{\grammian{A}{B}{\horizon}}}} \\
		& \leq \sqrt{\inprod{I_{\sysdim}}{I_{\sysdim}}} = \sqrt{\sysdim}.
	\end{align*} 
	Since inequality \eqref{eq:moq prop 1} is simply the Cauchy-Schwartz inequality applied to the matrices $\grammian{A}{B}{\horizon}$ and $I_{\sysdim}$, equality holds if and only if $\grammian{A}{B}{\horizon} = \lambda I_{\sysdim}$, and this happens exactly when the frame $\ftf$ is tight. This proves the property \ref{p:moq properties 1}. 

	At this point we take a slight detour to observe that inequality \eqref{eq:moq prop 1} can be stated in a more general sense. Let $\ftf$ be a sequence of square summable elements from an $n$ dimensional Hilbert space $H_n$, not necessarily constituting a frame. From \cite[Lemma 1, p.\ 7]{casazza2008classes} we see that
	\begin{equation}
		\label{eq: nfp moq relation}
		\begin{aligned}
			\trace\bigl(G_{ \ftf }^2\bigr) &= \sum_{i = 1}^{+\infty} \sum_{j = 1}^{+\infty} \abs{\inprod{\ltiframeelement{i}}{\ltiframeelement{j}}}^2, \quad \text{and} \\
			\trace\bigl(G_{ \ftf }\bigr)   &= \sum_{i = 1}^{+\infty} \norm{v_i}^2.
	\end{aligned}
	\end{equation}
	We then define the Normalized Frame Potential (NFP) of the frame $\ftf$ by
	\begin{equation}
		\label{eq: nfp def}
		\NFP\big(\ftf\big) \Let \frac{\sum_{i = 1}^{+\infty} \sum_{j = 1}^{+\infty} \abs{\inprod{\ltiframeelement{i}}{\ltiframeelement{j}}}^2}{\big(\sum_{i = 1}^{+\infty} \norm{\ltiframeelement{i}}^2\big)^2}.
	\end{equation}
	Observe that by \eqref{eq: nfp moq relation} and \eqref{eq: nfp def},
	\begin{equation}
		\label{eq: nfp moq relation 2}
		\frac{\trace{\bigl(G_{\ftf}\bigr)}}{\sqrt{\trace{\bigl(G_{\ftf}^2}\bigr)}} = \frac{1}{\sqrt{\NFP\big(\ftf\big)}},
	\end{equation}
	and this along with \eqref{eq:moq prop 1} tells us that
	\begin{equation}
		\label{eq: nfp ineq}
		\NFP\big(\ftf\big) \geq \frac{1}{{n}}.
	\end{equation}
	a fact that has been proved independently in \cite{casazza2006physical}. We emphasize here that inequality \eqref{eq: nfp ineq} states that the NFP of a sequence of vectors is greater than or equal to the inverse of the dimension of the Hilbert space it belongs to. The Normalized Frame Potential has a clear physical interpretation outlined in \cite{casazza2006physical} and \cite{benedetto2003finite}.
	
	We shall now prove the property \ref{p:moq properties 2} by contradiction: Suppose that \ref{p:moq properties 2} is false, that is, assume that $\moq{A}{B}{\horizon} > \sqrt{d}$ for some positive integer \(d\le \sysdim\), and that the dimension of the reachable subspace at time $\horizon$ is at most equal to $d$. This means that $\dim\big(\Span\ftf\big) \leq d$, where $\ftf$ is the frame generated by the LTI system \eqref{eq:lti-system} according to \eqref{eq: lti frame}. However, $\Span\ftf$ endowed with the inner product induced from $\R{\sysdim}$ is a Hilbert space in its own right, and $\ftf$ is a subset of this Hilbert space. By its definition in Equation \eqref{eq: nfp def}, the \(\NFP\) of $\ftf$ as a subset of $\R{\sysdim}$ and that as a subset of $\Span\ftf$ are equal. Applying the inequality \eqref{eq: nfp ineq} to $\ftf$ as a subset of $\Span\ftf$, we get
	\[
		\NFP\big(\ftf\big) \geq \frac{1}{\dim\big(\Span\ftf\big)} \geq \frac{1}{d},
	\]
	and along with \eqref{eq: nfp moq relation 2} this implies that
	\[
		\moq{A}{B}{\horizon} \leq \sqrt{d}.
	\]
	This is a contradiction of our starting hypothesis, which proves the property \ref{p:moq properties 2}.
\end{proof}

We now discuss some features of the frame-theoretic MOQ that distinguish it from the three classical MOQs.

\begin{remark}
	When the LTI system \eqref{eq:lti-system} is uncontrollable, both $\trace \big( \congrammian^{-1} \big)$ and $\lambda_{\min}^{-1} \big( \congrammian \big)$ are undefined, and $\det \big( \congrammian \big)$ is zero. This means that as far as the three classical MOQs are concerned, all uncontrollable systems are equally bad. On the one hand and in contrast to the classical MOQs, the frame-theoretic MOQ \(\eta\) has the unique ability to distinguish between uncontrollable systems. In addition, the property \ref{p:moq properties 2} of Proposition \ref{p: moq properties} says that if the frame-theoretic MOQ is larger than a certain value, the dimension of the reachable subspace is guaranteed to be larger than a precise corresponding value. Therefore, increasing the frame-theoretic MOQ leads to an increase in the rank of the reachable subspace even if the system under consideration is uncontrollable. On the other hand, one drawback of the frame-theoretic MOQ \(\eta\) is that it cannot distinguish between controllable and uncontrollable systems completely. Indeed, even though the property \ref{p:moq properties 2} of Proposition \ref{p: moq properties} says that if $\moq{A}{B}{\horizon}$ is greater than $\sqrt{\sysdim - 1}$, then the system is controllable, it is possible that the system is controllable but $\moq{A}{B}{\horizon} \leq \sqrt{\sysdim - 1}$.
\end{remark}

\begin{remark}
	All three classical MOQs are increasingly hard to compute as the dimension \(\sysdim\) of the state space increases. However, once the Gramian $\grammian{A}{B}{\horizon}$ is known, evaluating the frame-theoretic MOQ \(\eta\) involves very little computation. Indeed, since $\grammian{A}{B}{\horizon}$ is symmetric, $\trace{\bigl(\grammian{A}{B}{\horizon}^2\bigr)}$ is simply the sum of squares of entries of $\grammian{A}{B}{\horizon}$ and therefore, computing this involves just $\sysdim^2$ multiplications and as many additions. Computing $\trace{\bigl(\grammian{A}{B}{\horizon}\bigr)}$ involves summing $\sysdim$ diagonal entries of $\grammian{A}{B}{\horizon}$. In total, consequently, all one needs is $\sysdim^2$ number of multiplications, $(\sysdim^2 + \sysdim)$ number of additions, and one square-root operation to compute the frame-theoretic MOQ. Of course, computing the Gramian itself requires the evaluation of the integral \eqref{eq:control-grammian}, which may be difficult for large scale systems. However, in the case where the system matrix $A$ is asymptotically stable, the infinite horizon Gramian 
	\begin{equation}\label{eq: inf grammian}
		G_{(A, B)} \Let \int_{0}^{+\infty} \exp{tA}BB\transp\exp{t A\transp} \,\dd t
	\end{equation}
	is well-defined, and can be computed easily by solving the Lyapunov equation
	\begin{equation}\label{eq: lyapunov cont}
		AG_{(A, B)} + G_{(A, B)}A\transp + BB\transp = 0.
	\end{equation}
	There are efficient numerical algorithms available for solving the Lyapunov equation, even for large scale systems.
\end{remark}

\begin{remark}
	We mention that quantity $\trace{\bigl(\grammian{A}{B}{\horizon}\bigr)}$, which is similar to our frame-theoretic MOQ, has appeared in the literature before; see, e.g., \cite{pasqualetti2014controllability, ref:SumCorLyg-16, ref:ZhaCor-17}. However, in these articles the usage of $\trace{\bigl(\grammian{A}{B}{\horizon}\bigr)}$ as a measure of quality was mainly motivated by two factors:
	\begin{enumerate}[label=\textup{(\roman*)}, leftmargin=*, widest=ii, align=right]
		\item the observation that
			\begin{equation}
				\label{eq:trace moq prop}
				\frac{\trace{\bigl(\grammian{A}{B}{\horizon}^{-1}\bigr)}}{\sysdim} \geq \frac{\sysdim}{\trace{\bigl(\grammian{A}{B}{\horizon}\bigr)}},
			\end{equation}
			and
		\item the fact that $\trace{\bigl(\grammian{A}{B}{\horizon}\bigr)}$ is submodular as a function of the columns of the $B$ matrix.
	\end{enumerate}
	The inequality \eqref{eq:trace moq prop} suggests that $\trace{\bigl(\grammian{A}{B}{\horizon}\bigr)}$ is inversely related to $\trace{\bigl(\grammian{A}{B}{\horizon}^{-1}\bigr)}$; increasing $\trace{\bigl(\grammian{A}{B}{\horizon}\bigr)}$ can potentially lead to a decrease in $\trace{\bigl(\grammian{A}{B}{\horizon}^{-1}\bigr)}$. This is, however, only a heuristic, and it has been observed that increasing the trace of the Gramian need not ensure controllability and often leads to a poor choice of system (see, e.g., \cite[Section 5]{pasqualetti2014controllability}).
\end{remark}

\begin{remark}\label{rem: submod}
	Often an engineer is faced with the task of selecting the columns of the matrix \(B\) of a particular system with fixed system matrix $A$ from a given finite set of vectors. In this situation, it is desirable to do so by maximizing some measure of quality of the resulting system. Since the set of possible choices is finite, this leads to a combinatorial optimization problem, and for large scale systems such problems may not be tractable. In such cases, submodularity \footnote{Let $V$ be a given finite set. A set function \(f: 2^V \rightarrow \R{}\) is called submodular if for all subsets $A \subset B \subset V$ and all elements $s \notin B$, it holds that $f(A \cup \{s\}) - f(A) \geq f(B \cup \{s\}) - f(B)$. In the context of the discussion in Remark \ref{rem: submod}, the set $V$ is the finite set of possible choices for columns of the matrix $B$, and the MOQ is called submodular if it is submodular as a function on the finite set of choices of columns for the $B$ matrix.}  is a property that plays an important role in combinatorial optimization similar to that of convexity in continuous optimization. In the presence of the property of submodularity, there exists efficient numerical algorithms with proven performance guarantees that can solve large scale combinatorial optimization problems. We refer the reader to \cite{ref:SumCorLyg-16} for more details on submodularity and its consequences on measures of quality. We mention here that the frame-theoretic MOQ does not posses the property of submodularity.
\end{remark}

\begin{remark}
	All of the theory developed in this article, including the measure of quality defined in this section can verbatim be extended to the case of discrete time LTI systems of the form 
	\[
		{}x(t+1) = A x(t) + B u(t) \quad \text{for } t = 0, 1, 2, \ldots
	\]
	In fact, we initiated the study of frame-theoretic measures of quality in the context of discrete time linear systems in \cite{ref:SheCha-17CDC}. In this case, the set of control profiles at horizon $\horizon\in\N$ is the set of $\horizon$-tuples of elements in $\R{\condim}$, identified with $\mat{\condim}{\horizon}$. One can endow this vector space of control profile with an inner product similar to that defined in \eqref{eq: inner prod} by replacing the integral with a sum. The minimum control effort problem can now be defined analogously and the definition of the controllability Gramian changes to
	\[
		\grammian{A}{B}{\horizon} \Let \sum_{t = 0}^{\horizon - 1} A^t B B\transp (A\transp)^t.
	\]
	In the case of stable systems, the infinite horizon Gramian is well-defined as
	\[
		G_{(A,B)} \Let \sum_{t = 0}^{+\infty} A^t B B\transp (A\transp)^t.
	\]
	The infinite horizon Gramian can be computed with relative ease as the solution of the equation
	\[
		G_{(A, B)} - AG_{(A, B)}A\transp = BB\transp.
	\]
\end{remark}


\bibliographystyle{amsalpha}
\bibliography{ref}

\end{document}